\documentclass[preprint,a4paper,USenglish,cleveref, autoref,thm-restate]{lipics-v2021}

\usepackage{mathtools}
\usepackage{cleveref}

\newcommand{\tw}{{\mathtt{tw}}}
\newcommand{\pw}{{\mathtt{pw}}}

\newcommand{\vc}{{\mathtt{vc}}}
\newcommand{\ml}{{\mathtt{ml}}}
\newcommand{\vi}{{\mathtt{vi}}}
\newcommand{\td}{{\mathtt{td}}}
\newcommand{\inc}{{\mathtt{inc}}}
\newcommand{\connE}{{\mathtt{connE}}}
\newcommand{\STNT}{\textsc{Non-Terminal Spanning Tree}}
\newcommand{\MSTNT}{\textsc{Minimum Weight Non-Terminal Spanning Tree}}
\newcommand{\revised}[1]{\textcolor{black}{#1}}
\newcommand{\NT}{V_{\rm NT}}

\title{Finding a Minimum Spanning Tree with a Small Non-Terminal Set}

\author{Tesshu Hanaka}{Kyushu University, Japan}{johnqpublic@dummyuni.org}{https://orcid.org/0000-0001-6943-856X}{JSPS KAKENHI Grant Number JP21K17707, JP21H05852, JP22H00513, JP23H04388, and JST CRONOS Grant Number JPMJCS24K2}

\author{Yasuaki Kobayashi}{Hokkaido University, Japan}{koba@ist.hokudai.ac.jp}{https://orcid.org/0000-0003-3244-6915}{JSPS KAKENHI Grant Numbers JP20H00595 and JP23H03344}

\authorrunning{T. Hanaka and Y. Kobayashi} 

\Copyright{Tesshu Hanaka and Yasuaki Kobayashi} 

\ccsdesc[500]{Theory of computation~Parameterized complexity and exact algorithms} 

\keywords{Spanning tree, Fixed-parameter algorithms, Parameterized complexity, Kernelization} 

\nolinenumbers 

\begin{document}

\maketitle

\begin{abstract}
In this paper, we study the problem of finding a minimum weight spanning tree that contains each vertex in a given subset $\NT$ of vertices as an internal vertex. This problem, called {\MSTNT}, includes \textsc{$s$-$t$ Hamiltonian Path} as a special case, and hence it is NP-hard.
In this paper, we first observe that {\STNT}, the unweighted counterpart of {\MSTNT}, is already NP-hard on some special graph classes. Moreover, it is W[1]-hard when parameterized by clique-width. In contrast, we give a $3k$-vertex kernel and $O^*(2^k)$-time algorithm, where $k$ is the size of non-terminal set $\NT$.
The latter algorithm can be extended to {\MSTNT} with the restriction that each edge has a polynomially bounded integral weight.
We also show that {\MSTNT} is fixed-parameter tractable parameterized by the number of edges in the subgraph induced by the non-terminal set $\NT$, extending the fixed-parameter tractability of {\MSTNT} to a more general case.
Finally, we give several results for structural parameterization.
\end{abstract}

\section{Introduction}
The notion of spanning trees plays a fundamental role in graph theory, and the minimum weight spanning tree problem is arguably one of the most well-studied combinatorial optimization problems.
Numerous variants of this problem are studied in the literature (e.g. \cite{Goemans2006,HASSIN1995,Rosamond2016}).
In this paper, we consider {\MSTNT}, which is defined as follows. 
In a (spanning) tree $T$, leaf vertices, which are of degree $1$, are called \emph{terminals}, and other vertices are called \emph{non-terminals}. 
In {\MSTNT}, we are given an edge-weighted graph $G=(V, E, w)$ with $w\colon E \to \mathbb R^+$ and a set of designated non-terminals $\NT\subseteq V$. The goal of this problem is to find a minimum weight spanning tree $T$ of $G$ subject to the condition that each vertex in $\NT$ is of degree at least $2$ in $T$.
We also consider the unweighted variant of {\MSTNT}, which we call \STNT: Given a graph $G = (V, E)$ and $\NT \subseteq V$, the goal is to determine whether $G$ has a spanning tree $T$ such that each vertex in $\NT$ is of degree at least $2$ in $T$.

{\MSTNT} was firstly introduced by Zhang and Yin \cite{ZHANG2012}. Unlike the minimum weight spanning tree problem, {\MSTNT} is NP-hard~\cite{ZHANG2012}.
Nakayama and Masuyama \cite{Nakayama2016:cograph} observed that {\STNT} is also NP-hard as \textsc{$s$-$t$ Hamiltonian Path}, the problem of finding a Hamiltonian path between specified vertices \revised{$s$ and $t$}, is a special case of {\STNT}: When $\NT = V\setminus \{s,t\}$, any solution of {\STNT} is a Hamiltonian path between $s$ and $t$.
Nakayama and Masuyama~\cite{Nakayama2016:cograph,Nakayama2017:outerplanar,Nakayama2018:interval,Nakayama2019:SP} devised polynomial-time algorithms for {\STNT} on several classes of graphs, such as cographs, outerplanar graphs, and series-parallel graphs.

\subsection{Our contribution}\label{sec:contribution}
In this paper, we study {\MSTNT} and {\STNT} from the viewpoint of parameterized complexity.

\revised{On the positive side,} we \revised{first} give a $3k$-vertex kernel for {\STNT}, where $k$ is the number of vertices in $\NT$.
We also give polynomial kernelizations for {\STNT} with respect to vertex cover number and max leaf number.
For {\MSTNT}, we give an $O^*(2^k)$-time algorithm\footnote{The $O^*$ notation suppresses a polynomial factor of the input size.} when the weight of edges is integral and upper bounded by a polynomial in the input size.  Moreover, we design an $O^*(2^{\ell})$-time algorithm for {\MSTNT} for an arbitrary edge weight, where $\ell$ is the number of edges in the subgraph induced by $\NT$. Note that this parameter $\ell$ is ``smaller'' than  $k$ in the sense that $\ell \le \binom{k}{2}$ but $k$ could be arbitrarily large even when $\ell=0$.
Finally, we show that the property of being a spanning tree that has all vertices in $\NT$ as internal vertices is expressed by an MSO$_2$ formula, which implies that {\MSTNT} is fixed-parameter tractable when parameterized by treewidth.

On the negative side, we observe that {\STNT} is NP-hard even on planar bipartite graphs of maximum degree 3, strongly chordal split graphs, and chordal bipartite graphs. Moreover, we observe that {\STNT} is W[1]-hard when parameterized by cliquewidth, and it cannot be solved in time $2^{o(n)}$ unless the Exponential Time Hypothesis fails\revised{, where $n$ is the number of vertices in the input graph}.
In contrast to the fact that {\STNT} admits a polynomial kernelization with respect to vertex cover number, we show that it does not admit a polynomial kernelization with respect to vertex integrity under some complexity-theoretic assumption.

    

\subsection{Related work}
The minimum weight spanning tree problem is a fundamental graph problem.
There are many variants of this problem, such as \textsc{Minimum Diameter Spanning Tree}~\cite{HASSIN1995}, \textsc{Degree-Bounded Spanning Tree}~\cite{Goemans2006}, \textsc{Diameter-Bounded Spanning Tree}~\cite{GareyJ79}, and \textsc{Max Leaf Spanning Tree}~\cite{Rosamond2016}.

\textsc{Max Internal Spanning Tree} is highly related to {\STNT}. Given a graph $G$ and an integer $k$, the task of \textsc{Max Internal Spanning Tree} is to find a spanning tree $T$ of $G$ that has at least $k$ internal vertices. The difference between {\STNT} and \textsc{Max Internal Spanning Tree} is that the former designates internal vertices in advance whereas the latter does not.
Since \textsc{Max Internal Spanning Tree} is a generalization of \textsc{Hamiltonian Path}, it is NP-hard.
For restricted graph classes, the problem can be solved in polynomial time, such as interval graphs \cite{LI2018_MIST_interval},  cacti, block graphs, cographs, and bipartite permutation graphs \cite{sharma2021_internal_graphclass}.
Binkele-Raible et al.~\cite{Binkele2013_MIST_Exact} give an $O^*(3^n)$-time algorithm for \textsc{Max Internal Spanning Tree} and then Nederlof~\cite{Nederlof_MIST_3nExact} improves the running time to $O^*(2^n)$, where $n$ is the number of vertices of the input graph.
For the fixed-parameter tractability, Fomin et al.~\cite{FOMIN2013_MIST_3k_kernel} show that \textsc{Max Internal Spanning Tree} admits a $3k$-vertex kernel. Then Li et al.~\cite{LI2017_MIST_2k-kernel} improve the size of the kernel by giving a $2k$-vertex kernel. 
The best known deterministic $O^*(4^k)$-time algorithm is obtained by combining a $2k$-vertex kernel with an $O^*(2^n)$-time algorithm, while there is a randomized $O^*(\min \{3.455^k, 1.946^n\})$-time algorithm \cite{Bjorklund2017_MIST_FPT}.
For approximation algorithms, Li et al.~\cite{LI2021_MIST_approx} propose a $3/4$-approximation algorithm and Chen et al.~\cite{Chen2018_MIST_approx} improve the approximation ratio to $13/17$.

\textsc{Max Leaf Spanning Tree} is a ``dual'' problem of \textsc{Max Internal Spanning Tree}, where the goal is to compute a spanning tree that has leaves as many as possible. This problem is highly related to \textsc{Connected Dominating Set} and well-studied in the context of parameterized complexity \cite{CyganBook2015,Rosamond2016}.

The \emph{designated leaves} version of \textsc{Max Leaf Spanning Tree} is called \textsc{Terminal Spanning Tree}. In this problem, given a graph $G=(V,E)$ and terminal set $W\subseteq V$, the task is to find a spanning tree $T$ such that every vertex in $W$ is \revised{a} leaf in $T$. Unlike \textsc{Non-Terminal Spanning Tree}, it can be solved in polynomial time by finding a spanning tree in the subgraph induced by $V\setminus W$ and adding vertices in $W$ as leaves.

\revised{In the context of graph theory, {\STNT} has already discussed in \cite{EgawaO15,Kiraly18}. More precisely, the paper \cite{EgawaO15} gives a sufficient condition that for a graph $G$, a vertex set $S$, and a function $f\colon S \to \mathbb N$, $G$ has a spanning tree such that the degree of each vertex $v \in S$ is at least $f(v)$ in the spanning tree. Kir\'{a}ly \cite{Kiraly18} gives the shorter proofs of the sufficient condition and proposes a polynomial-time algorithm for checking the condition.}

\section{Preliminaries}\label{sec:pre}
In this paper, we use several basic concepts and terminology in parameterized complexity.
We refer the reader to \cite{CyganBook2015}. 

\medskip
\noindent \textbf{Graphs.} Let $G = (V, E)$ be an undirected graph.
We denote by $V(G)$ and $E(G)$ the \revised{vertex set and edge set} of $G$, respectively.
We use $n$ to denote the number of vertices in $G$.
For $X \subseteq V$, we denote by $G[X]$ the subgraph induced by $X$.
For $v\in V$ and $X \subseteq V$, $G-v$ and $G-X$ denote $G[V\setminus \{v\}]$ and  $G[V\setminus X]$, respectively. 
We also define $G-e=(V,E\setminus \{e\})$ and $G+e=(V,E\cup\{e\})$.
For $v\in V$, we let $N_G(v)$ denote the set of neighbors of $v$ and $N_G[v]=N_G(v)\cup \{v\}$. This notation is extended to sets: for $X \subseteq V$, we let $N_G[X]=\bigcup_{v\in X}N_G[v]$ and $N_G(X)=N_G[X]\setminus X$.
When the subscript $G$ is clear from the context, we may omit it.

A \emph{forest} is a graph having no cycles. If a forest is connected, it is called a \emph{tree}.
In a forest $F$, a vertex is called a \emph{leaf} if it is of degree 1 in $F$, and called an \emph{internal} vertex otherwise.

\begin{definition}
For a vertex set $X \subseteq V$ and a forest $F$, we say $F$ is \emph{admissible} for $X$ if each $v\in X$ is an internal vertex in $F$.

\end{definition}
In {\STNT}, given a non-terminal set $\NT$,  the goal is to determine whether there is an admissible spanning tree for $\NT$ in $G$.
Throughout the paper, we denote $k=|\NT|$ and $\ell=|E(G[\NT])|$, and assume that $k \le n - 2$.
We also define an optimization version of {\STNT}. In {\MSTNT}, we are additionally given an edge weight function $w\colon E(G) \to \mathbb R$, the goal is to find an admissible spanning tree $T$ for $\NT$ in $G$ minimizing its total weight $w(E(T)) = \sum_{e \in E(T)} w(e)$.
We also assume that the input graph $G$ is connected as otherwise, our problems are trivially infeasible.
The following easy proposition is useful to extend a forest into a spanning tree.

\begin{proposition}\label{prop:forest_to_SPtree}
Let $G$ be a connected graph and $F$ be a forest in $G$.
Then there is a spanning tree that contains all edges in $F$.
\end{proposition}

\noindent \textbf{Graph parameters.}
A set $S$ of vertices is called a \emph{vertex cover} if every edge has at least one endpoint in $S$. The vertex cover number $\vc(G)$ of $G$ is defined by the size of a minimum vertex cover in $G$. The \emph{vertex integrity} $\vi(G)$ of $G$ is the minimum integer \revised{$p$} satisfying that there exists $S\subseteq V$ such that \revised{$|S| + \max_{H\in \texttt{cc}(G-S)}|V(H)|\le p$}, where \revised{$\texttt{cc}(G-S)$ is the set of connected components of $G-S$}. The \emph{max leaf number} $\ml(G)$ of $G$ is the maximum integer \revised{$q$} such that there exists a spanning tree having \revised{$q$} leaves in $G$. For the definition of treewidth $\tw(G)$, pathwidth $\pw(G)$, and treedepth $\td(G)$, we refer the reader to the books~\cite{Nesetril:Sparsity:2012,CyganBook2015}.
It is well-known that these graph parameters have the following relationship.

\begin{proposition}[\cite{GraphParam,FLMMRS2009:maxleaf,GimaHKKO22}]
    For every graph $G$, it holds that $\tw(G)\le \pw(G)\le \td(G)-1\le \vi(G)-1\le \vc(G)$ and $\tw(G)\le \pw(G)\le 2\ml(G)$.
\end{proposition}

\medskip

\noindent \textbf{Matroids.} Let $U$ be a finite set.
A pair $\mathcal M = (U, \mathcal B)$ with $\mathcal B \subseteq 2^U$ is called a \emph{matroid}\footnote{This definition is equivalent to that defined by the family of \emph{independent sets} \cite{Oxley_Matroid_2006}.} if the following axioms are satisfied:
\begin{itemize}
    \item $\mathcal B$ is nonempty;
    \item For $X, Y \in \mathcal B$ with $X \neq Y$ and for $x \in X \setminus Y$, there is $y \in Y \setminus X$ such that $(X \setminus \{x\}) \cup \{y\} \in \mathcal B$.
\end{itemize}
It is not hard to verify that all the sets in $\mathcal B$ have the same cardinality.

There are many combinatorial objects that can be represented as matroids.
Let $H = (V, E)$ be a graph.
If $\mathcal B_H$ consists of all subsets of edges, each of which forms a spanning tree in $H$, then the pair $(E, \mathcal B_H)$ is a matroid, which is called a \emph{graphic matroid}.
We are also interested in another matroid.
Let $U$ be a finite set and let \revised{$\{U_1, U_2, \ldots, U_t\}$} be a partition of $U$: $U_i \cap U_j = \emptyset$ for $i \neq j$ and \revised{$\bigcup_{1 \le i \le t}U_i = U$}.
For \revised{$1 \le i \le t$}, let $\ell_i$ and $u_i$ be two non-negative integers with $\ell_i \le u_i$.
For a non-negative integer $r$, we define a set $\mathcal B_r$ consisting of all size-$r$ subsets $U' \subseteq U$ such that $\ell_i \le |U' \cap U_i| \le u_i$ for \revised{$1 \le i \le t$}.
Then, the pair $(U, \mathcal B_r)$ is a matroid unless $\mathcal B_r = \emptyset$.

\begin{proposition}\label{prop:LUBmatroid}
    If $\mathcal B_r \neq \emptyset$, then $(U, \mathcal B_r)$ is a matroid.
\end{proposition}
\begin{proof}
    It suffices to show that $\mathcal B_r$ satisfies the second axiom of matroids.
    Let $X, Y \in \mathcal B_r$ with $X \neq Y$.
    For \revised{$1 \le i \le t$}, let $x_i = |X \cap U_i|$ (resp. $y_i = |Y \cap U_i|$).
    Let $x \in X \setminus Y$ and assume that $x \in U_i$.
    Suppose first that $x_i \le y_i$.
    Since $x \in X \setminus Y$ and $|X \cap U_i| \le |Y \cap U_i|$, there is $y \in U_i$ such that $y \in Y \setminus X$.
    Then, we have $Z = (X \setminus \{x\}) \cup \{y\} \in \mathcal B_r$ as $|Z \cap U_j| = |X \cap U_j|$ for \revised{$1 \le j \le t$}.
    Suppose next that $x_i > y_i$.
    As $\sum_{1 \le j \le k}x_j = \sum_{1 \le j \le k}y_j = r$, there is an index $i'$ with $x_{i'} < y_{i'}$.
    This implies that $U_{i'}$ contains an element $y \in Y \setminus X$.
    Then, $Z = (X \setminus \{x\}) \cup \{y\} \in \mathcal B_r$ as $\ell_i \le y_i \le x_i - 1 = |Z \cap U_i|$, $|Z \cap U_{i'}| = x_{i'}+1 \le y_{i'} \le u_{i'}$, and $|Z \cap U_{j}| = |X \cap U_j|$ for all $j' \neq i, i'$.
\end{proof}


Let $\mathcal M_1 = (U, \mathcal B_1)$ and $\mathcal M_2 = (U, \mathcal B_2)$ be matroids.
It is well known that there is a polynomial-time algorithm to check whether $\mathcal B_1 \cap \mathcal B_2$ is empty when set $U$ and an \emph{independence oracle} is given as input (e.g., \cite{Frank1981}), where an independence oracle for a matroid $\mathcal M = (U, \mathcal B)$ is a black-box procedure that given a set $U' \subseteq U$, returns true if there is a set $B \in \mathcal B$ that contains $U'$.
Moreover, we can find a minimum weight common base in two matroids in polynomial time.
\begin{theorem}[\cite{Frank1981}]\label{thm;matroid_intersect}
    Let $U$ be a finite set and let $w\colon U \to \mathbb R$.
    Given two matroids $\mathcal M_1 = (U, \mathcal B_1)$ and $\mathcal M_2 = (U, \mathcal B_2)$, we can compute a set $X \in \mathcal B_1 \cap \mathcal B_2$ minimizing $w(X)$ in polynomial time, provided that the polynomial-time independence oracles of $\mathcal M_1$ and $\mathcal M_2$ are given as input.
\end{theorem}

Finally, we observe that there are polynomial-time independence oracles for graphic matroids and $(U, \mathcal B_r)$.
By Proposition~\ref{prop:forest_to_SPtree}, the independence oracle for the graphic matroid defined by a connected graph $G$ returns true for given $F \subseteq E(G)$ if and only if $F$ is a forest.
The following lemma gives an independence oracle for $(U, \mathcal B_r)$.

\begin{lemma}\label{lem:ind-oracle}
    There is a polynomial-time algorithm that given a set $U' \subseteq U$, determines whether there is a set $U'' \in \mathcal B_r$ with $U' \subseteq U'' \subseteq U$, that is, $|U''| = r$ and $\ell_i \le |U'' \cap U_i| \le u_i$ for all $1 \le i \le k$.
\end{lemma}

\begin{proof}
    The problem can be reduced to that of finding a degree-constrained subgraph of an auxiliary bipartite graph.
    To this end, we construct a bipartite graph $H$ as follows.
    The graph $H$ consists of two independent sets $A$ and $B$, where each vertex $a_i$ of $A$ corresponds to a block $U_i$ of the partition and each vertex $b_e$ of $B$ corresponds to an element $e$ in $U \setminus U'$.
    For each $e \in U \setminus U'$, $H$ contains an edge between $b_e$ and $a_i$, where $i$ is the unique index with $e \in U_i$.
    Then, we define functions $d^\ell \colon A \cup B \to \mathbb N$ and $d^u \colon A \cup B \to \mathbb N$ as:
    \begin{align*}
        d^\ell(v) = \begin{cases}
            0 & \text{ if } v \in B\\
            \ell_i - |U' \cap U_i| & \text{ if } v \in A
        \end{cases}\hspace{0.4cm}
        d^u(v) = \begin{cases}
            1 & \text{if } v \in B\\
            u_i - |U' \cap U_i| & \text{if } v \in A
        \end{cases}
    \end{align*}
    Then, we claim that there is a set $U''$ satisfying the condition in the statement of this lemma if and only if $H$ has a subgraph $H'$ such that $d^\ell(v) \le d_{H'}(v) \le d^{u}(v)$ for all $v \in V(H)$, where $d_{H'}(v)$ is the degree of $v$ in $H'$, and $H'$ contains exactly $r - |U'|$ edges.
    From a feasible set $U'' \subseteq U$, we construct a subgraph $H'$ in such a way that for each $b_e \in B$ we take the unique edge $(a_i, b_e)$ if $e \in U'' \setminus U'$.
    Clearly, this subgraph $H'$ contains $r - |U'|$ edges.
    Moreover, $H'$ satisfies the degree condition as 
    \begin{align*}
    d_{H'}(a_i) = |(U'' \setminus U') \cap U_i| = |U'' \cap U_i| - |U' \cap U_i| \ge \ell_i - |U' \cap U_i| = d^\ell(a_i),
    \end{align*}
    and, analogously, we have $d_{H'}(a_i) \le d^u(a_i)$.
    This transformation is reversible and hence the converse direction is omitted here.
    
    Given a graph $H$, an integer $r'$, and the degree bounds $d^\ell, d^u$, there is a polynomial-time algorithm that finds a subgraph $H'$ with $r'$ edges satisfying $d^\ell(v) \le d_{H'}(v) \le d^u(v)$ for $v \in V(H)$~\cite{Gabow83:degree-subgraph,Shiloach81:degree-subgraph}.
    \revised{Hence}, the lemma follows.
\end{proof}

\section{Kernelization}
\subsection{Linear kernel for $k$}

In this subsection, we give a linear vertex kernel of {\STNT} when parameterized by $k = |\NT|$.
Let $G$ be a graph with a non-terminal set $\NT \subseteq V(G)$.
The following easy lemma is a key to our results.

\begin{lemma}\label{lem:contraction}
There is an admissible spanning tree for $\NT$ in $G$ if and only if there is an admissible forest for $\NT$ in $G[N[\NT]]$.
\end{lemma}
\begin{proof}
Let $T$ be an admissible spanning tree for $\NT$ in $G$.
We remove vertices in $V(G) \setminus N[\NT]$ from $T$.
Since this does not change the degree of any vertex in $\NT$, we have an admissible forest in $G[N[\NT]]$.

Conversely, suppose that there is an admissible forest $F$ for $\NT$ in $G[N[\NT]]$. This forest is indeed a forest in $G$. By Proposition \ref{prop:forest_to_SPtree}, there is a spanning tree of $G$ that contains $F$ as a subgraph. This spanning tree is admissible for $\NT$ in $G$ as the degree of any vertex in $\NT$ is at least $2$ in $F$.
\end{proof}

Let $G'$ be the graph obtained from $G$ by deleting all the vertices in $V(G)\setminus N[\NT]$, adding a vertex $r$, and connecting $r$ to all the vertices in $N(\NT)$.
From the assumption that $G$ is connected, $G'$ is also connected.
By Lemma \ref{lem:contraction}, we can immediately obtain the following corollary.

\begin{corollary}\label{cor:contraction}
There is an admissible spanning tree for $\NT$ in $G$ if and only if there is an admissible spanning tree for $\NT$ in $G'$.
\end{corollary}

%


Due to Corollary \ref{cor:contraction}, we can ``safely'' reduce $G$ to $G'$\revised{, that is, $G$ has an admissible spanning tree for $\NT$ if and only if $G'$ does.}
The graph $G'$ may still have an arbitrary number of vertices as the size of $N_{G'}(\NT)$ cannot be upper bounded by a function in $k$.
To reduce this part, we apply the expansion lemma \cite{Fomin_Hitting_2016}.

\begin{definition}\label{def:q-expansion}
Let $H=(A\cup B, \revised{E_H})$ be a bipartite graph. For positive integer $q$, $M\subseteq \revised{E_H}$ is called a \emph{$q$-expansion} of $A$ into $B$ if it satisfies the following:
\begin{itemize}
    \item every vertex in $A$ is incident to exactly $q$ edges in $M$.
    \item there are exactly $q|A|$ vertices in $B$, each of which is incident to exactly one edge in $M$.
\end{itemize}
\end{definition}

\begin{lemma}[Expansion lemma \cite{Fomin_Hitting_2016,Thomasse10}]\label{lem:expansion}
Let $q$ be a positive integer and $H=(A\cup B, \revised{E_H})$ be a bipartite graph such that:
\begin{itemize}
    \item $|B|\ge q|A|$, and
    \item there \revised{are no isolated vertices} in $B$.
\end{itemize}
Then there are non-empty vertex \revised{sets} $X\subseteq A$ and $Y\subseteq B$ such that:
\begin{itemize}
    \item $N_H(Y)\subseteq X$, and
    \item there is a $q$-expansion $M \subseteq \revised{E_H}$ of $X$ into $Y$.
\end{itemize}
Moreover, one can find such $X$, $Y$, and $M$ in polynomial time in the size of $H$.
\end{lemma}

\revised{Suppose that $N_{G'}(\NT)$ has at least $2k$ vertices.}
Let $H = (\NT \cup N_{G'}(\NT), \revised{E_H})$ be the bipartite graph obtained from $G'[N_{G'}[\NT]]$ by deleting all edges between vertices in $\NT$ and those between vertices in $N_{G'}(\NT)$.
As $N_{G'}(\NT)$ has no isolated vertices in $H$, by Lemma~\ref{lem:expansion}, there exists a $2$-expansion $M \subseteq \revised{E_H}$ of $X$ into $Y$ for some non-empty \revised{sets} $X \subseteq \NT$ and $Y \subseteq N_{G'}(\NT)$ such that $N_H(Y) \subseteq X$.
We then construct a smaller instance $(\hat{G}, \hat{V}_{\rm NT})$ of ${\STNT}$ as follows.
We first delete all vertices in $Y$ from $G'$ and add an edge between $r$ and $x$ for each $x \in X$.
The graph obtained in this way is denoted by $\hat{G}$.
Then we set $\hat{V}_{\rm NT} = \NT \setminus X$.
See \Cref{fig:reduction} for an illustration.

\begin{figure}[t]
    \centering
    \includegraphics[width=\textwidth]{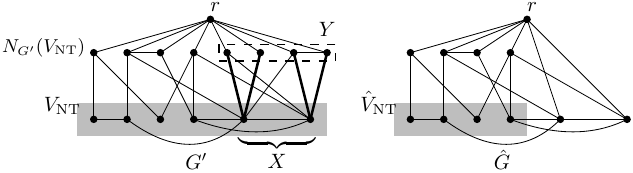}
    \caption{The construction of $\hat{G}$. The shaded areas indicate the vertices of $\NT$ and $\hat{V}_{\rm NT}$.
    The bold lines are edges in a $2$-expansion of $X$ into $Y$.}
    \label{fig:reduction}
\end{figure}

Observe that $\hat{G}$ is also connected.
This follows from the fact that every vertex in $(N_{G'}(\NT) \cup X) \setminus Y$ is adjacent to $r$ and every vertex $v$ in $\hat{V}_{\rm NT}$ is adjacent to some vertex of $(N_{G'}(\NT) \cup X) \setminus Y$ in $\hat{G}$ as otherwise $v \in \hat{V}_{\rm NT}$ is adjacent to a vertex in $Y$, which violates the fact that $N_H(Y) \subseteq X$.

\begin{lemma}\label{lem:deletion}
There exists an admissible spanning tree for $\NT$ in $G'$ if and only if there exists an admissible spanning tree for $\hat{V}_{\rm NT}$ in $\hat{G}$.
\end{lemma}

\begin{proof}
Suppose that there exists an admissible spanning tree $T$ for $\NT$ in $G'$.
Since $T$ is admissible for $\NT$, for every $v \in \NT \setminus X$, $T$ has at least two edges incident to $v$.
These edges are contained in $\hat{G}$ as there are no edges between $Y$ and $\NT \setminus X$.
Thus, there is an admissible forest for $\hat{V}_{\rm NT}$ in $\hat{G}$\revised{, consisting of these two incident edges for each $v \in \NT \setminus X$}, and by Proposition~\ref{prop:forest_to_SPtree}, $\hat{G}$ has an admissible spanning tree for $\hat{V}_{\rm NT}$.

Conversely, let $\hat{T}$ be an admissible spanning tree for $\hat{V}_{\rm NT}$ in $\hat{G}$.
To construct a spanning tree of $G'$, we select all edges of $\hat{T}$ incident to $\hat{V}_{\rm NT}$. 
These edges are \revised{also} contained in $G'$.
Moreover, we select all edges in the 2-expansion $M$ of $X$ into $Y$.
Since there are no edges between $\hat{V}_{\rm NT}$ and $Y$ in $G'$, the subgraph consisting of selected edges does not have cycles.
Moreover, as every vertex of $\NT$ is of degree at least $2$ in the subgraph, it is an admissible forest for $\NT$ in $G'$\revised{.}
By Proposition \ref{prop:forest_to_SPtree}, there is an admissible spanning tree for $\NT$ in $G'$.
\end{proof}


Our kernelization is described as follows.
From a connected graph $G$, we first construct the graph $G'$ by deleting the vertices in $V(G) \setminus N[\NT]$ and adding $r$ adjacent to each vertex in $N(\NT)$.
By Corollary~\ref{cor:contraction}, $(G, \NT)$ is a yes-instance of {\STNT} if and only if so is $(G', \NT)$.
If $G'$ has at most $3k$ vertices, we are done.
Suppose otherwise.
As $V(G') \setminus N_{G'}[\NT] = \{r\}$ and $|\NT| = k$, we have $|N_{G'}(\NT)| \ge 2k$.
We apply the expansion lemma to $G'$ and then obtain $\hat{G}$ and $\hat{V}_{\rm NT}$.
By Lemma~\ref{lem:deletion}, $(G, \NT)$ is a yes-instance of {\STNT} if and only if so is $(\hat{G}, \hat{V}_{\rm NT})$.
We repeatedly apply these reduction rules as long as the reduced graph $\hat{G}$ has at least $3|\hat{V}_{\rm NT}| + 1$ vertices.
Therefore, we have the following theorem.
\begin{theorem}\label{thm:kernel:k}
{\STNT} admits a $3k$-vertex kernel.
\end{theorem}

\subsection{Linear kernel for vertex cover number}
In this subsection, we show that $\STNT$ admits a linear kernel when parameterized by vertex cover number.
We first apply an analogous transformation used in the previous subsection.
By~\Cref{lem:contraction}, there is an admissible spanning tree for $\NT$ in $G$ if and only if there is an admissible forest for $\NT$ in $G[N[\NT]]$.
We remove all vertices of $V(G) \setminus N[\NT]$ and then add a vertex $r$ that is adjacent to every vertex in $N(\NT)$.
As seen in the previous subsection, $G$ has an admissible spanning tree for $\NT$ if and only if the obtained graph has.
Thus, in the following, $G$ consists of three vertex sets $\NT$, $N(\NT)$, and $\{r\}$, where $r$ is adjacent to every vertex in $N(\NT)$. 
\revised{Furthermore, we delete all the edges within  $G[N(\NT)]$, which is safe by Lemma~\ref{lem:contraction}. Thus, $N(\NT)$ forms an independent set in $G$.}
Let $S$ be a vertex cover of $G$ and $I=V\setminus S$. 
As $G$ is obtained from a subgraph of the original graph by adding a vertex $r$ \revised{and deleting edges within $G[N(\NT)]$}, we have $|S| \le \tau + 1$, where $\tau$ is the vertex cover number of the original graph.
 Suppose that $|\NT\cap I|\ge |S|$. Then we conclude that $G$ has no admissible spanning tree for $\NT$. 

\begin{claim}\label{claim:vc:1}
If $|\NT\cap I|\ge |S|$, there is no admissible spanning tree for $\NT$. 
\end{claim}
\begin{proof}
Suppose that $T$ is an admissible spanning tree for $\NT$ in $G$. Then $T$ has at least $2|\NT\cap I|$ edges between $\NT\cap I$ and $S$. Consider the subforest $T'$ of $T$ induced by $(\NT\cap I)\cup S$. We have 
\begin{align*}
    2|\NT\cap I|\le |E(T')|\le |\NT\cap I| + |S| - 1.
\end{align*}
Thus, we have $|\NT\cap I|\le |S|-1$. 
\end{proof}

Suppose next that $2|\NT \cap S| \le |N(\NT)\cap I|$.
We define a bipartite graph $H = ((\NT \cap S) \cup (N(\NT) \cap I), \revised{E_H})$, where $\revised{E_H}$ is the set of edges in $G$, each of which has an endpoint in $\NT \cap S$ and the other in $N(\NT) \cap I$.
As $2|\NT \cap S| \le |N(\NT)\cap I|$ and each vertex in $N(\NT)\cap I$ has a neighbor in $\NT \cap S$, there is a $2$-expansion of $X \subseteq \NT \cap S$ into $Y \subseteq N(\NT) \cap I$ in $H$.
By applying the same reduction rule used in \Cref{lem:deletion}, the obtained instance $(\hat{G}, \hat{V}_{\rm NT})$ satisfies the following claim.
\begin{claim}\label{claim:vc:2}
Suppose that $G$ satisfies $2|\NT \cap S| \le |N_{G}(\NT)\cap I|$. Then $G$ has an admissible spanning tree for $\NT$ if and only if $\hat{G}$ has an admissible spanning tree for $\hat{V}_{\rm NT}$.
\end{claim}

The proof of this claim is analogous to that in \Cref{lem:deletion}.
Let us note that $S$ remains a vertex cover of $\hat{G}$.

Our kernelization is formalized as follows.
Let $G$ be an input graph.
Let $S$ be a vertex cover of $G$ with $|S| \le 2\cdot\vc(G)$ and let $I = V(G) \setminus S$.
Such a vertex cover $S$ can be computed in polynomial time by a well-known $2$-approximation algorithm for the minimum vertex cover problem.
Then our kernel $(G', \NT')$ is obtained by
exhaustively applying these reduction rules by Claims \ref{claim:vc:1} and \ref{claim:vc:2} to $G$.
Let $S' = (V(G') \cap S) \cup \{r\}$ be a vertex cover of $G'$ and let $I' = V(G') \setminus S'$.
After these reductions, we have 
\begin{align*}
    |V(G')| &= |S'| + |I'| &&\\
    &\le |S'| + |\NT' \cap I'| + |N(\NT')\cap I'| &&\\
    &\le |S'| + |\NT'\cap I'| + 2|\NT\cap S'| - 1 && (\text{by \Cref{claim:vc:2}})\\
    &\le |S'| + |S'|-1  + 2|\NT\cap S'| - 1&& (\text{by \Cref{claim:vc:1}})\\
    &= 4|S'| - 2&& (\text{by $|S'|=|S|+1$})\\
    &= 4|S|+2,
\end{align*}
yielding the following theorem.


\begin{theorem}
{\STNT}  admits a $(8\vc(G)+2)$-vertex kernel.
\end{theorem}

\subsection{Quadratic kernel for max leaf number}

We show that $\STNT$ admits a polynomial kernel when parameterized by max leaf number.
Thanks to the following lemma, we can suppose that the input graph has  at most  $4\ml(G) - 2$ vertices of degree at least 3.

\begin{lemma}[\cite{KW1991:maxleaf,FLMMRS2009:maxleaf,BJK2013:maxleaf}]\label{lem:max_leaf}
Every graph $G$ with max leaf number $\ml(G)$ is a subdivision of a graph with at most $4\ml(G) - 2$ vertices. In particular, $G$ has at most $4\ml(G) - 2$ vertices of degree at least 3.
\end{lemma}

Thus, we only have to reduce the number of vertices of degree at most~2. Let $V_i\subseteq V$ be the set of vertices of degree $i$ and $V_{\ge i} = \bigcup_{j\ge i} V_j$.
In the following, we say that a reduction rule is \emph{safe} if the instance $(G', \NT')$ obtained by applying the reduction rule is equivalent to the original instance $(G, \NT)$: $G$ has an admissible spanning tree for $\NT$ if and only if $G'$ has an admissible spanning tree for $\NT'$.

\revised{In the following, we assume that each vertex in $\NT$ has degree at least~2, as otherwise the instance is clearly infeasible. Let $v$ be a vertex of degree~1. By the assumption, we have $v \notin \NT$.
If $v$ has a neighbor that is not in $\NT$, we can safely delete $v$. Otherwise, $v$ has a neighbor $u$, which belongs to $\NT$}. Then we delete $v$ and set $\NT:=\NT\setminus \{u\}$. This is safe because \revised{the degree of $u$ is at least $2$ in any spanning tree of $G$.}

Suppose that $G$ has two adjacent vertices $u,v\in V_2$ of degree $2$. Let $x$ (resp. $y$) be the other neighbor of $u$ (resp. $v$). If both $u$ and $v$ are contained in $\NT$, we contract edge $e = \{u,v\}$ and include the corresponding vertex $uv$ in $\NT$. Since $\{x,u\},\{u,v\},\{v,y\}$ must be contained in any admissible spanning tree for $\NT$, this reduction is also safe.
Suppose that \revised{neither of $u$ and $v$ is contained in $\NT$}.
If $\{u,v\}$ is a bridge in $G$, both $\{x, u\}$ and $\{v, y\}$ are also bridges in $G$, meaning that every admissible spanning tree for $\NT$ in $G$ contains all of these bridges. Thus, we can contract edge $e = \{u, v\}$, which is safe.
Otherwise, $G-e$ remains connected. If $G$ has an admissible spanning tree $T$ for $\NT$ containing $e$, forest $T - e$ is admissible for $\NT$.
Then, by~\Cref{prop:forest_to_SPtree}, there is an admissible spanning tree for $\NT$ in $G - e$. Thus, we can safely remove edge $e$ from $G$. In $G-e$, the degrees of $u$ and $v$ are exactly one. Therefore, we can further remove $u,v$ safely. 


From the above reductions, we can assume that $G$ does not have  
consecutive vertices $u,v$ of degree 2, both of which are in $\NT$ or not in $\NT$. 
Let $p,q,r,s$ be four consecutive vertices of degree 2 in $G$. Let $x$ (resp. $y$) be the other neighbor of $p$ (resp. $s$).
Without loss of generality, \revised{we assume that} $p,r\in \NT$ and $q,s\notin \NT$. 
Let $G'$ be the graph obtained from $G$ by contracting edge $\{p,q\}$ to a vertex $pq$ and $V'_{NT} = \NT \setminus \{p,q\}\cup \{pq\}$. 
Since any admissible spanning tree for $\NT$ in $G$ contains all the edges $\{x, p\}$, $\{p, q\}$, and $\{q, r\}$, we can reduce the instance $(G,\NT)$ to $(G',V'_{NT})$ safely.

Hence, when the above reductions are applied exhaustively, $G$ has neither a vertex of degree $1$ nor consecutive four vertices of degree $2$. 

Now we show that the reduced graph $G'$ has $O(\ml(G)^2)$ vertices.
As our reduction rules either delete a vertex of degree $1$, delete an edge between vertices of degree $2$, or contract an edge between vertices of degree $2$, they do not newly introduce a vertex of degree at least $3$.
This implies that the set of vertices of degree at least $3$ in $G'$ is a subset of $V_{\ge 3}$.
Moreover, as $G'$ has no vertices of degree $1$, we conclude that $G'$ is a subdivision of a graph $H$ with at most $4\ml(G) - 2$ vertices.
Since this subdivision is obtained from $H$ by subdividing each edge with at most three times, $G'$ contains $O(\ml(G)^2)$ vertices.

\begin{theorem}
{\STNT}  admits a quadratic vertex-kernel when parameterized by max leaf number.
\end{theorem}


\subsection{Kernel lower bound for vertex integrity}

In this subsection, we show that {\STNT} does not admit a polynomial kernel when parameterized by vertex integrity unless $\text{NP}\subseteq \text{coNP/poly}$.

\begin{theorem}\label{thm:kernel:vi}
{\STNT} does not admit a polynomial kernel when parameterized by vertex integrity unless $\text{NP}\subseteq \text{coNP/poly}$.
\end{theorem}

\begin{proof}
    We construct an AND-cross-composition \cite{BodlaenderJK14,Drucker12} from \textsc{$s$-$t$ Hamiltonian Path}. For $q$ instances $(G_i=(V_i,E_i), s_i, t_i)$ of \textsc{$s$-$t$ Hamiltonian Path} and a vertex $r$, we connect $r$ to $s_i$ by an edge $\{r,s_i\}$ for $1\le i\le q$. Let $G$ be the constructed graph and let $\NT = \bigcup_{1 \le i \le q} V_i\setminus \{t_i\}$. Then it is easy to see that every $G_i$ has an $s_i$-$t_i$ Hamiltonian path if and only if $G$ has an admissible spanning tree for $\NT$. Since the vertex integrity of $G$ is at most $\max_{1 \le i \le q} |V(G_i)| + 1$, the theorem holds.
\end{proof}

We remark that {\STNT} is significantly different from \textsc{$s$-$t$ Hamiltonian Path} with respect to kernelization complexity. 
\begin{remark}
\textsc{$s$-$t$ Hamiltonian Path} admits a polynomial kernel when parameterized by vertex integrity while it does not admit a polynomial kernel when parameterized by treedepth unless $\text{NP}\subseteq \text{coNP/poly}$.
\end{remark}

\begin{proof}
    We first show that \textsc{$s$-$t$ Hamiltonian Path} admits a polynomial kernel when parameterized by vertex integrity. Let $\vi(G)$ be the vertex integrity of $G$.
    \revised{By a polynomial-time $O(\log \vi(G))$ approximation algorithm~\cite{GimaHKM0O24}, we can obtain a vertex set $S$ such that $|S| + \max_{H\in \texttt{cc}(G-S)}|V(H)|= O(\vi(G)\log \vi(G))$ in polynomial time}. If the number of connected components of $G-S$ is at least \revised{$|S|+2$}, we immediately conclude that the instance is infeasible because any Hamiltonian path has to go through at least one vertex in $S$ from a connected component to another one. Otherwise, the number of connected components is at most \revised{$|S|+1$}. Since \revised{$|S| + \max_{H\in \texttt{cc}(G-S)}|V(H)|=O(\vi(G)\log \vi(G))$, the number of vertices in $G$ is at most $O(\vi(G)^2\log^2\vi(G))$}.

    To complement this positive result, we then show that the problem does not admit a polynomial kernel when parameterized by treedepth by showing an AND-cross-composition from \textsc{$s$-$t$ Hamiltonian Path}. For $q$ instances $(G_i=(V_i,E_i), s_i, t_i)$ of \textsc{$s$-$t$ Hamiltonian Path}, we connect  in series $G_i$ and $G_{i+1}$ by identifying $t_{i}$ as $s_{i+1}$ for $1\le i\le q-1$. Let $G$ be the constructed graph. By taking the center vertex (i.e., $s_{\lceil q/2\rceil + 1})$ of $G$ as a separator recursively, we can observe that the treedepth of $G$ is at most $\lceil \log_2 q\rceil + \max_{1 \le i \le q} |V(G_i)|$. It is to see that $G$ has an $s_1$-$t_q$ Hamiltonian path if and only if every $G_i$ has an $s_i$-$t_i$ Hamiltonian path. Thus,  the theorem holds.
\end{proof}

\section{Fixed-Parameter Algorithms} \label{sec:FPT}
\subsection{Parameterization by $k$}
By Theorem~\ref{thm:kernel:k}, {\STNT} is fixed-parameter tractable parameterized by $k = |\NT|$.
A trivial brute force algorithm on a $3k$-vertex kernel yields a running time bound $2^{O(k^2)} + n^{O(1)}$, where $n$ is the number of vertices in the input graph $G = (V, E)$.\footnote{A slightly non-trivial dynamic programming algorithm yields a better running time bound $O^*(8^k)$.}
However, this cannot be applied to the weighted case, namely {\MSTNT}.
In this subsection, we give an $O^*(2^k)$-time algorithm for {\MSTNT}, provided that the weight function $w$ is (positive) integral with $\max_{v\in V}w(v) = n^{O(1)}$.
The algorithm runs in (pseudo-)polynomial space.
Our algorithm is based on the Inclusion-Exclusion principle, which is quite useful to design exact exponential algorithms \cite{FominBook2010,CyganBook2015}, and counts the number of admissible spanning trees for $\NT$ in $G$.

\begin{theorem}[Inclusion-Exclusion principle]\label{thm:IE}
Let $U$ be a finite set and let $A_1,\ldots, A_t\subseteq U$.
Then, the following holds:
\begin{align*}
    |\bigcap_{i\in \{1,\ldots,t\}}A_i|=\sum_{X\subseteq \{1,\ldots, t\}}(-1)^{|X|}|\bigcap_{i\in X}\overline{A_i}|,
\end{align*}
where $\overline{A_i} = U \setminus A_i$ and $\bigcap_{i \in \emptyset} \overline{A_i} = U$.
\end{theorem}

Let $G = (V, E)$ be a graph with $w\colon E \to \mathbb N$ and let $W = \max_{v \in V} w(v)$.
For $0 \le q \le (n-1)W$, let $\mathcal T^q_G$ be the set of all spanning trees of $G$ with weight exactly $q$.
By a weighted counterpart of Kirchhoff's matrix tree theorem~\cite{Kirchhoff1847}, we can count the number of spanning trees in $\mathcal T^q_G$ efficiently.
\begin{theorem}[\cite{BroderM97:cnt-MWST}]\label{thm:weighted-matrix-tree-thm}
    There is an algorithm that, given an edge-weighted graph $G = (V, E)$ with $w\colon E \to \{1, 2, \ldots, W\}$ for some $W \in \mathbb N$ and an integer $q$, computes the number of spanning trees $T$ of $G$ with $w(T) = q$.
    Moreover, this algorithm runs in time $(n + W)^{O(1)}$ with space $(n + W)^{O(1)}$, where $n = |V|$.
\end{theorem}

For $v \in \NT$, let $\mathcal A^q_v\subseteq \mathcal T^q_G$ be the set of spanning trees of $G$ with weight $q$ that have $v$ as an internal node.
Clearly, the number of admissible spanning trees for $\NT$ with weight $q$ is $|\bigcap_{v\in \NT}\mathcal A^q_v|$.
Thus, we can solve {\MSTNT} by checking if $|\bigcap_{v\in \NT}\mathcal A^q_v| > 0$ for each $q$.
Due to Theorem \ref{thm:IE}, it suffices to compute $|\bigcap_{v\in X}\overline{\mathcal A^q_v}|$ for each $X \subseteq \NT$, where $\overline{\mathcal A^q_v} = \mathcal T^q_G \setminus \mathcal A^q_v$.
Here, $|\bigcap_{v\in X}\overline{\mathcal A^q_v}|$ is equal to the number of spanning trees $T \in \mathcal T^q_G$ such that every vertex in $X$ is a leaf of $T$.

\begin{lemma}\label{lem:counting}
Given $X \subseteq V$, we can compute $|\bigcap_{v \in X} \overline{\mathcal A^q_v}|$ in time $(n + W)^{O(1)}$.
\end{lemma}
\begin{proof}
Assume that $|V| \ge 3$.
Then, every leaf $w \in X$ of a spanning tree $T$ in $\bigcap_{v \in X} \overline{\mathcal A^q_v}$ has exactly one neighbor in $V \setminus X$.
Moreover, the subtree of $T$ obtained by removing all vertices in $X$ is a spanning tree of $G[V \setminus X]$.
Thus, the following equality holds:
\begin{align*}
    \left|\bigcap_{v \in X} \overline{\mathcal A^q_v}\right| = 
    \sum_{0 \le q' \le q} |\mathcal T^{q'}_{G[V\setminus X]}| \cdot |\mathcal M^{q-q'}(X, V\setminus X)|.
\end{align*}
\revised{In the above equality, for $Y \subseteq X$ and $0 \le j \le q$, we denote by $\mathcal M^{j}(Y, V\setminus X)$ the collection of edge subsets $M \subseteq E \cap (Y \times (V\setminus X))$ with $w(M) = j$ such that each vertex in $Y$ is incident to exactly one edge in $M$.}
By Theorem~\ref{thm:weighted-matrix-tree-thm}, we can compute $|\mathcal T^{q'}_{G[V \setminus X]}|$ in $(n + W)^{O(1)}$ time.
Thus, it suffices to compute $|\mathcal M^{q-q'}(X, V\setminus X)|$ in time $(n + W)^{O(1)}$ as well.
This can be done by dynamic programming described as follows.
Let $X = \{x_1, x_2, \ldots, x_p\}$.
For $0 \le i \le p$ and $0 \le j \le q - q'$, we define $m^j(i) = |\mathcal M^j(\{x_1, \ldots, x_{i}\}, V \setminus X)|$.
Clearly, $m^0(0) = 1$, $m^j(0) = 0$ for $j > 0$, and $m^j(p) = |\mathcal M^j(X, V \setminus X)|$ for $j \ge 0$.
For $i \ge 1$, it is easy to verify that
\begin{align*}
    m^j(i) = \sum_{e \in E \cap (\{x_i\} \times (V \setminus X))} m^{j - w(e)}(i - 1),
\end{align*}
where we define $m^{j'}(i-1) = 0$ for negative integer $j'$.
We can evaluate $m^j(i)$ in time $(n + W)^{O(1)}$ by dynamic programming, and hence the lemma follows.
\end{proof}

\begin{theorem}\label{thm:FPT}
{\MSTNT} is solvable in time $O^*(2^k)$ and polynomial space when \revised{the edge weight function} $w$ is integral with $\displaystyle\max_{v \in V} w(v) = n^{O(1)}$.
\end{theorem}

\subsection{Parameterization by $\ell$}
In this subsection, we give a fixed-parameter algorithm for {\MSTNT} with respect to the number of edges $\ell$ in $G[\NT]$.
Note that $\ell$ is a ``smaller'' parameter than $k$ in the sense that $\ell \le \binom{k}{2}$, while $k$ can be arbitrary large even if $\ell = 0$.
Note also that the algorithm described in this subsection works for {\MSTNT} with an arbitrary edge weight function, whereas the algorithm in the previous subsection works only for bounded integral weight functions.

We first consider the case where $\ell = 0$. 
To this end, we construct a partition of $E(G)$ as follows.
Let $\NT = \{v_1, v_2, \ldots, v_k\}$.
For each $v \in \NT$, we let $E_v$ be the set of edges that are incident to $v$ in $G$ and let $R = E(G) \setminus (E_{v_1}, E_{v_2},\ldots, E_{v_k})$.
Since $G[\NT]$ is edge-less, $\{E_{v_1}, E_{v_2}, \ldots, E_{v_k}, R\}$ is a partition of $E(G)$.
Clearly, a spanning tree $T$ of $G$ is admissible for $\NT$ if and only if $T$ contains at least two edges from $E_{v}$ for every $v \in \NT$.
This condition can be represented by the intersection of the following two matroids $\mathcal M_1$ and $\mathcal M_2$. 
\revised{$\mathcal M_1 = (E, \mathcal B_g)$} is just a graphic matroid of $G$ and $\mathcal M_2 = (E, \mathcal B_{n-1})$ is defined as follows: $\mathcal B_{n-1}$ consists of all edge \revised{subsets} $F \subseteq E(G)$ with $|F| = n - 1$ such that for $v \in \NT$, it holds that $|F \cap E_v| \ge \ell_v \coloneqq 2$.
Thus, a set of edges $F$ forms an admissible spanning tree for $\NT$ if and only if $F \in \mathcal B_g \cap \mathcal B_{n-1}$. 
By Proposition~\ref{prop:LUBmatroid}, $\mathcal M_2$ is a matroid.
By Theorem~\ref{thm;matroid_intersect} and Lemma~\ref{lem:ind-oracle}, we can find a minimum weight admissible spanning tree of $G$ for $\NT$ in polynomial time.

This polynomial-time algorithm can be extended to an $O^*(2^\ell)$-time algorithm for general $\ell \ge 0$.
For each $F \subseteq E(G[\NT])$, we compute a minimum weight admissible spanning tree $T$ for $\NT$ such that $E(T) \cap E(G[\NT]) = F$. 
To this end, we first check whether $F$ has no cycle.
If $F$ has a cycle, there is no such an admissible spanning tree for $\NT$.
Otherwise, we modify \revised{the matroids $\mathcal M_1$ and $\mathcal M_2$ as follows.
Let $G'$ be the multigraph obtained from $G$ by deleting edges in $E(G[\NT]) \setminus F$ and then contracting edges in $F$.
Note that the edge set of $G'$ corresponds to $E \setminus E(G[\NT])$ and hence $\mathcal E = \{E_{v_1} \setminus E(G[\NT]), \dots, E_{v_k} \setminus E(G[\NT]), R \setminus E(G[\NT])\}$ is a partition of $E(G')$.
Moreover, it is easy to see that every spanning tree of $G$ can be modified to a spanning tree of $G'$ by contracting all edges in $F$ and vice-versa.
Now, we define $\mathcal M'_1 = (E(G'), \mathcal B_{g'})$ as the graphic matroid of $G'$.
For $v \in \NT$, we let $\ell'_v = \max(0, 2 - d_{F}(v))$, where $d_{F}(v)$ is the number of edges incident to $v$ in $F$.
We denote by $\mathcal M'_2$ a pair $(E(G'), \mathcal B'_{n-|F|-1})$, where $\mathcal B'_{n-|F|-1}$ is the family of edge sets $F' \subseteq E(G')$ with $|F'| = n - |F| - 1$ such that for $v \in \NT$, it holds that $|F' \cap (E_v \setminus E(G[\NT]))| \ge \ell'_v$.
This pair is indeed a matroid as $\mathcal E$ is a partition of $E(G')$.}
Then, every spanning tree $T$ satisfying $E(T) \cap E(G[\NT]) = F$ is admissible for $\NT$ if and only if the edge set of $T$ belongs to the intersection of the following modified matroids $\mathcal M'_{1}$ and $\mathcal M'_2$.

$\mathcal M'_2$ is a matroid consisting a pair $(E(G) \setminus E(G[\NT]), \mathcal B'_{n-|F| -1})$, where $\mathcal B'_{n-|F|-1}$ is 
Thus, every common base in $\mathcal B'_g \cap \mathcal B'_{n-|F| - 1}$ corresponds to an edge set $F' \subseteq E(G) \setminus E(G[\NT])$ such that $F' \cup F$ forms an admissible spanning tree for $\NT$ in $G$.
Again, by~Theorem~\ref{thm;matroid_intersect} and Lemma~\ref{lem:ind-oracle}, we can compute a minimum weight admissible spanning tree $T$ for $\NT$ in $G$ with $E(T) \cap E(G[\NT]) = F$ in polynomial time, which yields the following theorem. 


\begin{theorem}\label{thm:FPT-ell}
    {\MSTNT} can be solved in time $O^*(2^{\ell})$ and polynomial space.
\end{theorem}

As a straightforward consequence of the above theorem, {\MSTNT} is fixed-parameter tractable parameterized by the number $k$ of non-terminals.

\begin{corollary}\label{cor:FPT-k:weighted}
    {\MSTNT} can be solved in time $O^*(2^{\binom{k}{2}})$ and polynomial space.
\end{corollary}\label{thm:FPT-k}

\subsection{Parameterization by tree-width}

Nakayama and Masuyama~\cite{Nakayama2017:outerplanar,Nakayama2019:SP} propose polynomial-time algorithms for {\STNT} on several \revised{subclasses} of bounded tree-width graphs, such as outerplanar graphs and series-parallel graphs. In this section, we show that {\STNT} can be solved in linear time on bounded tree-width graphs.

The property of being an admissible spanning tree for $\NT$ can be expressed by a formula in Monadic Second Order Logic, which will be discussed below.
By the celebrated work of Courcelle \cite{COURCELLE1990} and its optimization version \cite{ArnborgLS91:opt-Courcelle}, {\MSTNT} is fixed-parameter tractable when parameterized by tree-width.
\begin{theorem}
{\MSTNT} can be solved in linear time on bounded tree-width graphs.
\end{theorem}
\begin{proof}
Let $\texttt{NTST}(F)$ be an MSO$_2$ formula that is true if and only if $F \subseteq E(G)$ forms an admissible spanning tree for $\NT$.
$\texttt{NTST}(F)$ can be expressed as follows:
\begin{align*}
        \texttt{NTST}(F):=&\ \connE(F)\land \texttt{acyclic}(F) \land 
        _{\forall v\in V, \exists e\in F} \inc(v,e) \\
        &\land _{\forall v\in \NT, \exists e_1, e_2\in F}((e_1 \neq e_2) \land \inc(v,e_1)
        \land \inc(v,e_2) )\\
        \texttt{acyclic}(F) :=&
        \ _{\forall F' \subseteq F}( F' \neq \emptyset       \implies _{\exists v \in V} \texttt{deg1}(v, F')).
\end{align*}
$\texttt{acyclic}(F)$ is an auxiliary formula that is true if and only if $F$ is acyclic in $G$.
Here, $\connE(F)$ means that the subgraph induced by an edge set $F$ is connected and $\texttt{deg1}(v, F')$ means that the degree of vertex $v$ is exactly $1$ in the subgraph induced by an edge set $F'$,  which can be formulated in MSO$_2$ \cite{CyganBook2015}.
By the optimization version of Courcelle's theorem~\cite{ArnborgLS91:opt-Courcelle}, {\MSTNT} is fixed-parameter tractable when parameterized by tree-width.
\end{proof}

\section{Hardness results}


In this section, we observe that \textsc{$s$-$t$ Hamiltonian Path} is NP-hard even on several restricted classes of graphs, which immediately implies the NP-hardness of {\STNT} as well.
The results in this section follow from the following observation.
Let $G$ be a graph and let $v$ be an arbitrary vertex in $G$.
Let $G'$ be a graph obtained from $G$ by adding a new vertex $v'$ and an edge between $v'$ and $w$ for each $w \in N(v)$.
In other words, $v$ and $v'$ are false twins in $G'$.
Then, the following proposition is straightforward.
\begin{proposition}
    Suppose that $G$ has at least three vertices.
    Then $G$ has a Hamiltonian cycle if and only if $G'$ has a Hamiltonian path between $v$ and $v'$.
\end{proposition}
\begin{proof}
 Let $C$ be a Hamiltonian cycle of $G$ and let $w$ be one of the two vertices adjacent to $v$ in $C$.
 As $w$ is adjacent to $v'$ in $G'$, $C - \{v, w\} + \{v', w\}$ is a Hamiltonian path between $v$ and $v'$ in $G'$.
 Conversely, let $P$ be a Hamiltonian path between $v$ and $v'$ in $G'$.
 Let $w$ and $w'$ be the vertices of $G'$ that are adjacent to $v$ and $v'$ in $P$, respectively.
 As $P$ is a Hamiltonian path of length at least $3$, $w$ and $w'$ must be distinct.
 Moreover, both vertices are adjacent to $v$ in $G$, implying that $P - \{v', w'\} + \{v, w'\}$ is a Hamiltonian cycle of $G$.
\end{proof}

This proposition leads to polynomial-time reductions from \textsc{Hamiltonian Cycle} to \textsc{$s$-$t$ Hamiltonian Path} on several classes of graphs.
As \textsc{Hamiltonian Cycle} is NP-hard even on strongly chordal split graphs and chordal bipartite graphs~\cite{MULLER1996}, the following corollary follows.
\begin{corollary}\label{cor:st_HP}
\textsc{$s$-$t$ Hamiltonian Path} is NP-hard even on strongly chordal split graphs and chordal bipartite graphs.
\end{corollary}

Furthermore, since  $v$ and $v'$ are twins, the clique-width of $G'$ is the same as the clique-width of $G$.
As \textsc{Hamiltonian Cycle} is W[1]-hard when parameterized by clique-width~\cite{FominGLS10}, we obtain the following corollary.
\begin{corollary}\label{cor:st_HP_cliquewidth}
\textsc{$s$-$t$ Hamiltonian Path} is W[1]-hard when parameterized by clique-width.
\end{corollary}

In \cite{Melo_TCP_2021}, de Melo, de Figueiredo, and Souza show that \textsc{$s$-$t$ Hamiltonian Path} is NP-hard even on planar graphs of maximum degree 3.
The following theorem summarizes the above facts.
\begin{theorem}\label{thm:STNT:NP-hard}
{\STNT} is NP-hard even on planar bipartite  graphs of maximum degree 3, strongly chordal split graphs, and chordal bipartite graphs. Furthermore, it is W[1]-hard when parameterized by clique-width.
\end{theorem}

Furthermore, it is known that \textsc{Hamiltonian Cycle} cannot be solved in time $O^*(2^{o(n)})$ unless Exponential Time Hypothesis (ETH) fails \cite{CyganBook2015}. Thus, we immediately obtain the following theorem.
\begin{theorem}\label{thm:STNT:ETH}
{\STNT} cannot be solved in time $O^*(2^{o(n)})$ unless ETH fails.
\end{theorem}

\section{Conclusion}\label{sec:conclusion}
In this paper, we studied {\STNT} and {\MSTNT} from the viewpoint of parameterized complexity.
We showed that {\STNT} admits a linear vertex kernel with respect to the number of non-terminal vertices $k$, as well as polynomial kernels with respect to vertex cover number and max leaf number.
For the weighted counterpart, namely {\MSTNT}, we give an $O^*(2^k)$-time algorithm for graphs with polynomially-bounded integral edge weight and $O^*(2^{\ell})$-time algorithm for graphs with arbitrary edge weight, where $\ell$ is the number of edges in the subgraph induced by non-terminals.
We proved that {\MSTNT} is fixed-parameter tractable when parameterized by tree-width whereas it is W[1]-hard when parameterized by clique-width.

As future work, we are interested in whether {\STNT} can be solved in time  $O^*((2-\epsilon)^k)$ or $O^*((2-\epsilon)^{\ell})$ for some $\epsilon>0$. Also, it would be worth considering other structural parameterizations, such as cluster deletion number or modular-width.
\revised{Another possible direction is to consider the problem of finding a spanning tree that maximizes the number of internal vertices in $\NT$. This problem simultaneously generalizes our problem and \textsc{Max Internal Spanning Tree} by setting $V=\NT$.
It would be interesting to explore the (parameterized) approximability of this problem.
}



\bibliography{ref}

\end{document}